\documentclass[11pt, a4paper]{amsart}
\usepackage{verbatim}
\usepackage[dvips]{color}
\usepackage{amssymb}
\usepackage[active]{srcltx}
\usepackage{latexsym}
\usepackage{amsmath}
 \usepackage{float}
\usepackage{mathrsfs} 
\usepackage[draft]{fixme}
\usepackage{graphicx, epstopdf, float}

\usepackage{enumitem}

\usepackage{latexsym}

\usepackage{color}

\allowdisplaybreaks

\newtheorem{theorem}{Theorem}[section]
\newtheorem{proposition}{Proposition}[section]

\newtheorem{corollary}{Corollary}[section]
\theoremstyle{definition}

\newtheorem{remark}{Remark}[section]

\newtheorem{assumption}{Assumption}

\floatstyle{ruled} \restylefloat{figure} \restylefloat{table}

\numberwithin{equation}{section}

\newcommand{\beq}{\begin{equation}}
\newcommand{\bea}[1]{\begin{array}{#1} }
\newcommand{\eeq}{ \end{equation}}
\newcommand{\ea}{ \end{array}}

\def\mean#1{\mathchoice%
          {\mathop{\kern 0.2em\vrule width 0.6em height 0.69678ex depth -0.58065ex
                  \kern -0.8em \intop}\nolimits_{\kern -0.4em#1}}%
          {\mathop{\kern 0.1em\vrule width 0.5em height 0.69678ex depth -0.60387ex
                  \kern -0.6em \intop}\nolimits_{#1}}%
          {\mathop{\kern 0.1em\vrule width 0.5em height 0.69678ex
              depth -0.60387ex
                  \kern -0.6em \intop}\nolimits_{#1}}%
          {\mathop{\kern 0.1em\vrule width 0.5em height 0.69678ex depth -0.60387ex
                  \kern -0.6em \intop}\nolimits_{#1}}}

\def\vintslides_#1{\mathchoice%
          {\mathop{\kern 0.1em\vrule width 0.5em height 0.697ex depth -0.581ex
                  \kern -0.6em \intop}\nolimits_{\kern -0.4em#1}}%
          {\mathop{\kern 0.1em\vrule width 0.3em height 0.697ex depth -0.604ex
                  \kern -0.4em \intop}\nolimits_{#1}}%
          {\mathop{\kern 0.1em\vrule width 0.3em height 0.697ex depth -0.604ex
                  \kern -0.4em \intop}\nolimits_{#1}}%
          {\mathop{\kern 0.1em\vrule width 0.3em height 0.697ex depth -0.604ex
                  \kern -0.4em \intop}\nolimits_{#1}}}

\newcommand{\aveint}[2]{\mathchoice%
          {\mathop{\kern 0.2em\vrule width 0.6em height 0.69678ex depth -0.58065ex
                  \kern -0.8em \intop}\nolimits_{\kern -0.45em#1}^{#2}}%
          {\mathop{\kern 0.1em\vrule width 0.5em height 0.69678ex depth -0.60387ex
                  \kern -0.6em \intop}\nolimits_{#1}^{#2}}%
          {\mathop{\kern 0.1em\vrule width 0.5em height 0.69678ex depth -0.60387ex
                  \kern -0.6em \intop}\nolimits_{#1}^{#2}}%
          {\mathop{\kern 0.1em\vrule width 0.5em height 0.69678ex depth -0.60387ex
                  \kern -0.6em \intop}\nolimits_{#1}^{#2}}}

\def\eqn#1$$#2$${\begin{equation} \label#1#2\end{equation}}
\def\charfn_#1{{\raise1.2pt\hbox{$\chi
_{\kern-1pt\lower3pt\hbox{{$\scriptstyle#1$}}}$}}}

\def\qq1{q_*}
\def\q2{q_{**}}

\newdimen\vintbar
\def\vint{-\kern-\vintbar\int}

\def\0{\boldsymbol 0}

\newtoks\by
\newtoks\paper
\newtoks\book
\newtoks\jour
\newtoks\yr
\newtoks\pages
\newtoks\vol
\newtoks\publ

\def\name[#1, #2]{#1 #2}
\def\ota{{\hbox{\bf ???}}}
\def\cLear{\by = \ota\paper = \ota\book = \ota\jour = \ota\yr = \ota
\pages = \ota\vol = \ota\publ = \ota}
\def\endpaper{\the\by, \textit{\the\paper},
{\the\jour} \textbf{\the\vol} (\the\yr), \the\pages.\cLear}
\def\endbook{\the\by, \textit{\the\book},
\the\publ, \the\yr.\cLear}
\def\endpap{\the\by, \textit{\the\paper}, \the\jour.\cLear}
\def\endproc{\the\by, \textit{\the\paper}, \the\book, \the\publ,
\the\yr, \the\pages.\cLear}

\begin{document}


\title[ Moment Formula for  local/implied volatility ]{Moment Explosions, implied volatility and Local  Volatility at Extreme Strikes}



\author{ Sidi Mohamed Aly }
\address{Sidi Mohamed Aly\\ Centre for Mathematical Sciences, Mathematical Statistics\\
Lund University, Box 118, SE-221 00 Lund, Sweden}
\email{sidi@math.lth.se}

\begin{abstract}
We consider a stochastic volatility model where the moment generating function of the logarithmic price is finite only on part of the real line. Using a new Tauberian result obtained in \cite{Aly13} and \cite{Aly15}, we show that the knowledge of the moment generating function near its critical moment gives a sharp asymptotic expansion (with an error of order $o(1)$) of the local volatility and implied volatility for small and large strikes. We apply our theoretical estimates to Gatheral's SVI parametrization of the implied volatility and Heston's model.
 
 \medskip

\noindent

\medskip


\end{abstract}

\maketitle

\section{introduction}
\noindent In \cite{Aly13} and \cite{Aly15} we derived a Tauberian result that gives a sharp asymptotic formula for  the complementary cumulative distribution function of a random variable whose moment generation function (MGF) is finite only on part of the real line. The corresponding formula depends only on the behavior of the MGF near its critical moment. This result is very useful in the case where the moment generating function is known (as it is the case for the CIR process, Heston's model and many time changed L\'evy processes, cf. \cite{Aly13}, \cite{BenaimFriz08}). In \cite{Aly15} we have shown that this Tauberian result may also be used in cases where MGF is unknown in order to derive an asymptotic expansion of this MGF near its critical moment and then (using the theorem once more) to derive a sharp asymptotic formula for the cumulative distribution; indeed we proved in  \cite{Aly15}  that the MGF of $V^{2(1-p)}_t$, where $V$ is given by the stochastic differential equation
\[
d V_t = (a - b V_t ) d t + \sigma V^p_t dW_t,
\]
explodes at the critical moment $ \mu^\ast_t = \frac{ b }{ \sigma^2 (1-p)  ( 1 - e^{-2b(1-p) t}  )}_t $ and derived and sharp asymptotic formula for the MGF of $V^{2(1-p)}_t$ near $\mu^\ast_t$ as well as the complementary cumulative distribution function of $V^{2(1-p)}_t$.

In the present work we give another application to the Tauberian result of \cite{Aly15} to the well known Dupire's local volatility surface: $ \Sigma^2(t,k) = 2 \partial_T C (T,K)/K^2 \partial_{KK} C(T,K)|_{K=e^k}$ where $C(T,K)$ denotes the price of a European Call option with strike $K$ and maturity $T$. Indeed, if we denote $ X_t = \ln (S_t/S_0)$ with $S_t$ referring to the stock price, then the Call price is given as (assuming without loss of generality that the interest and dividend rates are 0) 
\[
C(T,K) = \mathbb{E} ( e^{X_T} - K)_+.
\]
Differentiating with respect to $K$ we have 
\[
\partial_K C(T,K) =- \mathbb{P} (X_T > \ln(K)).
\]
 On the other hand, our Tauberian result states that if the logarithm of the MGF: $\Lambda : p \to \ln  \mathbb{E} e^{ p X}$ is a regularly varying  function of $ \frac{1}{\mu^\ast - p} $ with a positive index near its critical moment  $\mu^\ast$  (plus other non restrictive assumptions) the  Fenchel-Legendre transform  $\Lambda^\ast$ of $\Lambda $ is  a good asymptotic approximation of the logarithm of  cumulative distribution of $X$ with $\limsup$-$\liminf$  arguments:
\begin{equation}\label{expans_cdf0}
\mathbb{P} (X \geq x) = e^{ - \Lambda^\ast (x)}  \left( \frac{ \sqrt{ { p^\ast}' (x) } }{ p^\ast (x) \sqrt{2 \pi}} - \frac{ 2 + \frac{\alpha}{ (\alpha + 1)^2 } }{  24 \mu^\ast  \sqrt{2 \pi} } ~\frac{1}{x^2 \sqrt{ {p^\ast}'(x) }  }    + o(     \frac{1}{x^2  \sqrt{ {p^\ast}' (x) } }    )   \right).
\end{equation}
Here $p^\ast$ is such that $\Lambda^\ast (x) = p^\ast (x) x + \Lambda (  p^\ast(x))$ and $\alpha$ is the index of the regularly varying function $ x \to \Lambda( \mu^\ast - \frac{1}{x}Ê)$.  In  particular if the logarithmic price $ X_t = \ln (S_t/S_0)$ satisfies the assumptions of the theorem then we immediately have a sharp asymptotic formula for the price of Call option as well as its derivative with respect to the strike; thus a sharp asymptotic formula can be derived  for Dupire's local volatility function $\Sigma(T,K)$ by integrating and differentiating (\ref{expans_cdf0}) with respect to $t$ and $x$. It will turn out in fact that a sharp asymptotic formula for local volatility function can be obtained in terms of $\Lambda$ and its derivates with respect to $t$ and $\mu$ via another method. We find indeed that under the same assumptions as for (\ref{expans_cdf0}), the local volatility for small/large strikes is given as
\begin{equation}\label{local_vol0}
  \Sigma^2(t, \pm y ) = \sigma^\pm_0 (t) y +
   \frac{  \partial_t \Delta_\pm (t,\nu_\pm(y)) \mp q \mp   \sigma^\pm_0 (t) ( 1 \mp 2 \mu^\ast_\pm   \pm \frac{1}{\nu_\pm (y) }Ê) y /\nu_\pm(y)     }{  \frac{1}{2}  ( (\mu^\ast_\pm   - \frac{1}{\nu_\pm(y)} )^2 \mp (\mu^\ast_\pm   - \frac{1}{\nu_\pm(y)} )) \left(   1 +  (\gamma^2_\pm - \gamma_\pm )\frac{ \nu_\pm(y)^2}{2 y^2 \nu'_\pm (y) }        \right) } + o(1),
\end{equation}
where $\sigma^\pm_0= \frac{- 2 \partial_t \mu^\ast_\pm (t)}{\mu^\ast_\pm (t) \mp \mu^\ast_\pm (t) } $, with $\mu^\ast_\pm (t)$ referring to the critical moment of $ (\pm X_t)$ and $ \nu_\pm$ is given in terms of the moment generating function of $(\pm X_t)$. In particular, the second term in the right hand side is  $ \mathcal{O} ( y^\frac{\alpha_\pm}{\alpha_\pm + 1} )$, where $\alpha_+$ (resp. $\alpha_-$) is the index of the regularly varying function $ \ln \mathbb{E} e^{ ( \mu^\ast_+ - \frac{1}{x} )X_t}$ (resp. $ \ln \mathbb{E} e^{ -( \mu^\ast_- - \frac{1}{x} )X_t}$).  

The expansion (\ref{local_vol0}) of the local volatility (and (\ref{expans_cdf0})) applies to a large class of time changed L\'evy models (see e.g. \cite{BenaimFriz08} and \cite{BenaimFriz09}). Several affine stochastic volatility models satisfy also the assumption of moment explosion (cf. \cite{KellerRessel11}). The most famous example of the later family is the Heston model, whose MGF satisfies the assumption of the main Tauberian theorem of \cite{Aly15}. In particular (\ref{expans_cdf0}) holds for the logarithmic price, as highlighted in \cite{Aly13}; this is a slight improvement on a result obtained in \cite{friz11}.  As regards the local volatility asymptotics, Friz and Gerhold show in \cite{friz15} that $ \Sigma (t, y) \sim \sigma^+_0 (t) y$. We clearly see that (\ref{local_vol0}) applied to Heston's model gives a substantial improvement on \cite{friz15}.  We draw attention that the authors of \cite{friz15} use Saddle point and Hankel contour integration methods to obtain an equivalence of the local volatility function for the specific Heston and NIG models, whereas our result is model independent.

Our Tauberian result also allows to derive a sharp asymptotic formula for the Black-Scholes implied volatility, similar to the one obtained for the local volatility function. Indeed, we show that the implied volatility for small/large stokes is given as function of $k:= \ln(K/S_0)$ by:
\begin{equation}\label{implied^0_vol0}
 t  \sigma^2 (t,\pm k)  = 4 \Lambda^\ast_\pm (t,k) + 2\tilde{c}^\pm_t (k) \mp 2 k - 4 \sqrt{ ( \Lambda^\ast_\pm(t,k) +\frac{  \tilde{c}^\pm_t (k)}{2})  (  \Lambda^\ast_\pm(t,k) +\frac{  \tilde{c}^\pm_t (k)}{2} \mp k)  }    + \mathcal{O} (k^\frac{-\alpha_\pm \wedge 1}{\alpha_\pm+1} ) .
\end{equation}
where  $ \Lambda^\ast_\pm (t,.)$ is Fenchel-Legendre transform of $\Lambda_\pm (t, .) $ and $\tilde{c}^\pm_t \sim \ln (k)$  is given in terms of $\Lambda$. We draw attention to the fact that under the theorem«s assumption, $ \Lambda^\ast_\pm (t,x) = \mu^\ast_\pm (t) x + \mathcal{O} ( x^\frac{\alpha_\pm}{\alpha_\pm + 1} )$. In particular, dividing both sides of (\ref{implied^0_vol0}) by $k$ and letting $k \to \pm \infty$ we find  Lee's moment formula \cite{lee04}.  

While our results apply to any model whose moment generating function of the logarithmic price is given explicitly, we choose to study two examples. In the first one, we consider the SVI parametrization of the implied volatility; we show that this will lead to an explosion of the MGF and that an asymptotic formula for the distribution tail as well as the local volatility may be easily obtained. The second example that we will study in details is the Heston model; we use the explicit formula for $\mathbb{E} S^p_t$, for $p < 0$ and $p \geq 1$ obtained in \cite{Aly13} (without any restriction on the parameters). We then apply our Tauberian results to derive sharp asymptotic expansions for the left/right wings of  local volatility as well as  implied volatility. We also show  that a similar formula holds for the Stein-Stein model.

The literature around the implied/volatility is very vast and the number of papers dealing with issues related to the implied/local volatility asymptotic in stochastic volatility model is countless. For the local volatility case, the most relevant work to this paper  is  \cite{friz15} mentioned above (see also \cite{stefano13} for the numerical and practical aspects of \cite{friz15}). As regards implied volatility and moment explosion, this work extends the famous Lee's moment formula obtained in \cite{lee04} by looking more closely to the moments of the stock prices. Benaim and Friz \cite{BenaimFriz08} have shown the equivalence of the regular variation property between the logarithm of density and the implied volatility. Other results deal with the implied volatility of some stochastic volatility models, such as \cite{friz11}, \cite{achil10}, \cite{achil14},  \cite{gatheral12}. 

This paper is organized as follows: In section~\ref{section1} we recall  the main Tauberian result of \cite{Aly15} that we will use to formulate and proof a new Tauberian result which will be the key to link the local volatility asymptotics to the MGF near its critical moment; this link will be given as a theorem that is formulated and proved in Section~\ref{section2}. In section~\ref{section3} we give a result relating the (sharp) asymptotics of the implied volatility for large strike and the MGF near its critical moment. Section~\ref{section4}  gives the application of our results to the SVI parametrization of the implied volatility, while Section~\ref{section5} studies extensively the application  the Heston model case. We then give the implication to Stein-Stein model. 

\section{Tauberian relation between the moment explosion and the distribution tails}\label{section1}
\noindent   Throughout this paper $(\omega, \mathcal{F}, (\mathcal{F}_t)_{t\geq 0}, \mathbb{P})$ is a complete filtered probability space satisfying the usual conditions and $\mathbb{E}$ refers to the expectation under $\mathbb{P}$.

In this section, we first  recall the Tauberian theorem that has been first  formulated in \cite{Aly13} and then extended in \cite{Aly15}. We then present a new Tauberian result that will be crucial when deriving the local volatility from the moment generating function.  We present the version that has been given in \cite{Aly15}; it concerns a random variable $X$ satisfying the following assumption:  

 \begin{assumption}\label{hyp_laregdev}
There exist $ \mu^\ast > 0$ and $ \alpha  > 0$ such that 
\begin{enumerate}[label=(\roman*)]
\item $\forall    \mu \in [0, \mu^\ast[, \; \Lambda(\mu) := \ln \mathbb{E} \; e^{\mu X} < + \infty $,  
\item The function $ f (x) ~:  x \longmapsto   \Lambda(\mu^\ast - \frac{1}{x}) $ is $ R_\alpha$ and   $\mathcal{C}^2 ([M, \infty[) $, for some $ M  $ sufficiently large.
\item The function $x \mapsto \frac{1}{\mu^\ast - p^\ast (x)}Ê$, with  $p^\ast (x) ={ \Lambda^\ast}' (x)$, is smoothly varying with index $\frac{1}{\alpha + 1}$, where $\Lambda^\ast$ is the Fenchel-Legendre transform of $\Lambda$.
\end{enumerate}
\end{assumption}
 
 We use the same notation as \cite{Aly15}; $ R_\alpha$ refers to the set of regularly varying function with index $\alpha$. It is also worth noticing that  (see Remark~2.1 in \cite{Aly15})  $\Lambda$ is convex. In particular, $\Lambda^\ast$ is well defined and is given as  $ \Lambda^\ast (x) = p^\ast(x) x - \Lambda ( p^\ast(x)  )$, with $p^\ast$ is the unique solution to $ \Lambda' (p^\ast(x)) = x$.
 
 \begin{theorem}\label{thm22_ext}
Let Assumption~\ref{hyp_laregdev} hold for some random variable $X$. Consider the function  $ \chi_x(.)$ defined, for $x$ sufficiently large,  by 
 \begin{equation}\label{psi_x_z}
 \chi_x (z) = (p^\ast(x)- p^\ast(xz )) x z + \Lambda(p^\ast(xz ))- \Lambda(p^\ast(x)),
 \end{equation}
 and let $g$ be smoothly varying with index $\gamma \in \mathbb{R}$. Then  for any $ \beta \in ]0,1[$, we have the following expansion as $x \to \infty$
\begin{equation}\label{taub_res}
 \int_{ x^{\beta-1} }^\infty g (xz) ~ e^{  \psi_x (z)} d z = \frac{ g(x) \sqrt{ \pi}}{ \sqrt{ 2  x^2  {p^\ast}'(x)}   } \left( 2 + \frac{ \gamma^2 + \frac{\gamma}{\alpha + 1}  + c_\alpha  }{x^2  {p^\ast}'(x) }  +  o ( \frac{1}{x^2 {p^\ast}'(x) }) \right),
\end{equation}
where
 \begin{eqnarray*}
c_\alpha &=&  - \frac{1}{4} (1 +  \frac{1}{\alpha + 1}) (2 +  \frac{1}{\alpha + 1}) + \frac{5}{12} ( 1 +  \frac{1}{\alpha + 1})^2.
\end{eqnarray*}
Furthermore, 
\begin{equation}\label{lim_sup2}
\limsup_{x \to \infty} Ê\left( x^2 \sqrt{ {p^\ast}' (x) } \Big[ \ln \mathbb{P} (X \geq x) + \Lambda^\ast (x) -\frac{ \sqrt{ { p^\ast}' (x) } }{ p^\ast (x) \sqrt{2 \pi}}   \Big] \right) \geq   - \frac{2 + \frac{\alpha}{ (\alpha + 1)^2 } }{24 \mu^\ast   \sqrt{2 \pi} } 
\end{equation}
and 
\begin{equation}\label{lim_inf2}
\liminf_{x \to \infty} ~Ê\left( x^2 \sqrt{ {p^\ast}' (x) } \Big[ \ln \mathbb{P} (X \geq x) + \Lambda^\ast (x) -\frac{ \sqrt{ { p^\ast}' (x) } }{ p^\ast (x) \sqrt{2 \pi}}   \Big] \right)\le  - \frac{2 + \frac{\alpha}{ (\alpha + 1)^2 } }{24 \mu^\ast    \sqrt{2 \pi} } . Ê
\end{equation}
In particular, if   "$\limsup$" equals "$\liminf$"  then we have 
\begin{equation}\label{expans_cdf}
\mathbb{P} (X \geq x) = e^{ - \Lambda^\ast (x)}  \left( \frac{ \sqrt{ { p^\ast}' (x) } }{ p^\ast (x) \sqrt{2 \pi}} - \frac{ 2 + \frac{\alpha}{ (\alpha + 1)^2 } }{  24 \mu^\ast  \sqrt{2 \pi} } ~\frac{1}{x^2 \sqrt{ {p^\ast}'(x) }  }    + o(     \frac{1}{x^2  \sqrt{ {p^\ast}' (x) } }    )   \right).
\end{equation}
\end{theorem}

The only use we will make of this theorem in this paper is to prove the following result which will be the key to link the local volatility asymptotics to the moment generating function near its critical moment. The proposition is followed by its proof.

\begin{proposition}\label{tauberian1}
Let $X$ be a random variable satisfying  Assumption~\ref{hyp_laregdev}. Assume that "$\limsup$" and "$\liminf$" in (\ref{lim_inf2}) and (\ref{lim_sup2}) are equal. Let $g$ be smoothly varying with index $Ê\gamma \in \mathbb{R}$ such that $ \mathbb{E}  | g(X)| < \infty$. ÊWe have, for $x$ sufficiently large,
\begin{equation}
\frac{ \mathbb{E} ( g(X) e^{( \mu^\ast - \frac{1}{x} ) X} )}{ \mathbb{E} (  e^{( \mu^\ast - \frac{1}{x} ) X} ) } = g( \Lambda' (\mu^\ast - \frac{1}{x}) )     \left( 1 + (\gamma^2 - \gamma) \frac{ \Lambda'' (\mu^\ast - \frac{1}{x}  )  }{  2 { \Lambda'}^2 ( \mu^\ast - \frac{1}{x}) } + o( \frac{ \Lambda'' (\mu^\ast - \frac{1}{x}  )  }{  2 { \Lambda'}^2 ( \mu^\ast - \frac{1}{x}) } Ê )   \right) .
\end{equation}
\end{proposition}

\begin{remark}
Note that under the theorem's assumption,  $ \frac{ \Lambda'' (\mu^\ast - \frac{1}{x}  )  }{  2 { \Lambda'}^2 ( \mu^\ast - \frac{1}{x}) }  = \mathcal{O} ( x^{- \alpha})$.
\end{remark}
\begin{proof}
We shall not assume that $X$ admits a smooth density; if that would be the case we could express  $ \mathbb{E} (U e^{px})$ by integrating with respect to this density function. Instead, we proceed as in \cite{Aly13} by using  the following representation of the exponential function:
\begin{eqnarray*}
g(U)  e^{p U} &=& g(U) 1_{ U \le c}  e^{p U} +  1_{ U > c} g(U)  e^{p U}  \\
&=& g(U) 1_{ U \le c}  e^{p U} + 1_{ U > c} \left(  g(c)  e^{ pc} +   \int_c^U e^{p z} (g'(z)  + p g(z) )  d z  \right)
\\ &=&
 g(U) 1_{ U \le c}  e^{p U} + 1_{ U > c}    g(c)  e^{ pc} + \int_c^\infty e^{p z} (g'(z)  + p g(z) ) 1_{z \le U} d z
 \\ &=&
 g(U \wedge c)    e^{p~ U \wedge c} + \int_c^\infty e^{p z} (g'(z)  + p g(z) ) 1_{z \le U} d z.
\end{eqnarray*}
which holds for any $ U ,  p, c \in \mathbb{R}$. Applying this to $U= X$  we have
\begin{eqnarray}\label{reprentation_exponential}
\mathbb{E} [ g(X) e^{p X} ]&=&   \mathbb{E}  [g(X \wedge c)    e^{p~ X \wedge c}]     + \int_c^\infty e^{p z} ( g'(z) + p g(z)) \mathbb{P}(z \le X) d z.
\end{eqnarray}
Let's denote, for a function $h$ with polynomial growth, 
\[
I^h (p) :=  \int_c^\infty  h(z) e^{p z}~ \mathbb{P}(z\le X ) d z.
\]
With this notation we can write $\mathbb{E} [ g(X) e^{p X} ] $ as
\begin{equation}\label{I_zeta}
\mathbb{E} [ g(X) e^{p X} ]  =   E_g(c)   +   I^{ g' } (p) + p I^{g  } (p),
\end{equation}
where 
\[
E_g (c) =   \mathbb{E}  [g(X \wedge c)    e^{p~ X \wedge c}]    .
\]
The idea here is to choose $c$ to be large but so that $ E_g (c)$ is negligible with respect to the two terms $ I^{ g' } $  and $I^{ g } $. It is worth noticing that under the assumption $ \mathbb{E} | g (X) |=  m \le \infty$, we have
\begin{equation}\label{uper_bound_Rg}
|E_g (c)| \le  \mathbb{E}  |g(X)|     e^{p c}  +  g (c) e^{Êpc} = ( m +  g (c) )e^{Êpc}.
\end{equation}

Now applying  Theorem~\ref{thm22_ext} to $X$ we have
  \[
  \mathbb{P}(z\le X ) = e^{ \Lambda^\ast (z) } f(z),
  \]
   where  $f$ is given as 
\[
f(z) = \frac{ \sqrt{ÊÊ{p^\ast }'(z)Ê} }{ p^\ast (z) \sqrt{2 \pi}} - \frac{   2 + \alpha/ (\alpha + 1)^2          }{ 24 p^\ast  (z) \sqrt{2\pi}} + o( \frac{1}{z^2 }  \sqrt{ÊÊ{p^\ast }'(z)Ê}).
\]
For $x$ sufficiently large, define $ Z \equiv Z(x)$ by
 \begin{equation}\label{Z_x}
 Z  = \Lambda' (t, \mu^\ast - \frac{1}{x} ) 
  \end{equation}
In particular, we have 
\[
 \mu^\ast - \frac{1}{x} = p^\ast  (Z).
\]
 It follows that, for any function $h$ (for which $I^h$ is well defined),   $ I^h ( \mu^\ast  -\frac{1}{x} ) \equiv I^h ( p^\ast (Z)) $ may be written as
\[
I^h( p^\ast (Z)) = \int_{c}^\infty h(z) e^{p^\ast (Z) } e^{ \Lambda^\ast (z) } f(z) d z = Z e^{\Lambda( p^\ast (Z))}
\int_{c/Z}^\infty  h(zZ)  f(zZ) e^{ \chi_x (z)} d z .
\]
We then choose $ c = Z^\beta$,  where $\beta = \frac{\alpha}{2(\alpha + 1)} ~~~( \in ]0,1[$). In particular we have $ Z^\beta \sim \sqrt{ \Lambda ( p^\ast  (Z) )}$ (and hence $ e^{ Z^\beta}$ is negligible with respect to $ e^{\Lambda (p^\ast  (Z)) } $).  Now we can easily see that $f \in R_\frac{-(\alpha + 2) }{ 2(\alpha + 1)}$. Hence $ g f \in R_{\gamma -\frac{\alpha + 2 }{ 2(\alpha + 1)} }$. The first statement of  Theorem~\ref{thm22_ext} ensures that
\[
I^g(p^\ast_t(Z)) =   e^{\Lambda( p^\ast (Z))}  \frac{  g(Z) f(Z) \sqrt{2 \pi}    }{ \sqrt{  {p^\ast}'(Z)  }  }  \left(
1 + \frac{  ( \gamma - \frac{\alpha + 2}{2(\alpha - 1)}  )^2  + ( \gamma - \frac{\alpha + 2}{2(\alpha + 1) }  )/(\alpha+1)   + c_\alpha     }{ Z^2 {p^\ast }'(Z)     }  + o(  \frac{1}{Z^2 {p^\ast }'(Z)} )
\right).
\]
A similar statement holds for $ I^{g'} (p^\ast (Z))$. It follows that for $x$ sufficiently large we have
\begin{equation}\label{ratio_Is}
\frac{I^g }{I^{ z \mapsto 1}} ( p^\ast  (Z)) = g(Z) \left( 1 + \frac{\gamma^2 - \gamma}{ 2 Z^2 { p^\ast}' (Z)} + o(\frac{1 }{Z^2 {p^\ast}' } ) \right)  , 
\end{equation}
and
\begin{equation}\label{ratio_Is2}
\frac{I^{g'} }{I^{ z \mapsto 1}} ( p^\ast  (Z)) = g'(Z) \left( 1 + \frac{ (\gamma-1)^2 - (\gamma-1)}{ 2 Z^2 { p^\ast}' (Z)} + o(\frac{1 }{Z^2 {p^\ast}' } ) \right)  , 
\end{equation}
where we used the fact that $g'$ is smoothly varying with index $\gamma - 1$. The notation $ z \mapsto 1$ refers to the "constant" function mapping $ \mathbb{R}$ into $ \{1\}$.

Now using the fact that $\beta$ is such that $ Z^\beta$ is negligible with respect to $ \Lambda (p^\ast (Z)) $ we have
\begin{eqnarray*}
\frac{  \mathbb{E} [ g(X) e^{  p^\ast (Z) X} ]     }{  \mathbb{E} [   e^{  p^\ast (Z) X} ]   }  
&=& 
 \frac{ E_g ( Z^\beta)  +   I^{ g' } + p^\ast (Z) I^g   }{  E_{z \mapsto 1 }  (Z^\beta) +  p^\ast (Z) I^{ z \mapsto 1}  } 
 \\ &=&
  \frac{ 1}{  p^\ast (Z)}  \frac{I^{g'} }{I^{ z \mapsto 1}} 
+   \frac{I^{ g } }{I^{z \mapsto 1}} + o ( e^{Z^{\beta} - \Lambda (p^\ast (Z))  }) ,
\end{eqnarray*}
where we used (\ref{uper_bound_Rg}) which ensures that $ E_g (Z^\beta)| \le ( m + g(Z^\beta) )  e^{ pZ^\beta } $.  In particular $ \frac{R_g (Z^\beta) }{ I^{ z \mapsto 1} }Ê= o ( e^{p Z^{\beta} - \Lambda (p^\ast (Z))  })$ decays exponentially to 0 as $Z$ goes to $\infty$ (recall that $ Z^{\beta}  \sim \sqrt{ \Lambda (p^\ast (Z) )}$). Hence
\[
\frac{  \mathbb{E} [ g(X) e^{ p^\ast (Z) X} ]     }{  \mathbb{E} [   e^{ p^\ast (Z) X} ]   }   = g(Z) \left( 1 + \frac{\gamma^2 - \gamma}{ 2 Z^2 { p^\ast}' (Z)} + o(\frac{1 }{Z^2 {p^\ast}' } ) \right).
\]
Now from the definition of $p^\ast$ we have $ \Lambda' (p^\ast (Z)) = Z$. It follows that
\[
{p^\ast}'(Z) = \frac{1}{\Lambda'' (p^\ast (Z))} = \frac{1}{\Lambda'' (\mu^\ast - \frac{1}{x}) } .
\]
On the other hand, $ {p^\ast}' \in R_{ -\frac{\alpha + 2}{\alpha + 1}  }$. It follows that
\[
  Z(x)^2 {p^\ast}' (Z(x))  \sim Z(x)^\frac{\alpha}{\alpha + 1} \sim x^{ - \alpha}.
\]
The  proof is completed by replacing $Z$ by $ \Lambda' (\mu^\ast - \frac{1}{x}) $ and ${p^\ast}' (Z)$ by $ 1/\Lambda'' (\mu^\ast - \frac{1}{x})$.
\end{proof}

\section{From moment explosion to the local volatility asymptotics}\label{section2}
\noindent  In this section we study the link between the local volatility asymptotics and and the moment explosion. We consider a stock price process $ S_t$  and  denote $ X_t := \ln (S_t/S_0)$.  Let's  also denote, for $t, p>0$,  
\begin{equation}
\Lambda_+ (t,p) = \ln \mathbb{E} ( e^{ p X_t} ) ~~~~~~~~~\textrm{and } ~~~~~~ \Lambda_- (t,p) = \ln \mathbb{E} ( e^{ -p X_t} ).
\end{equation}
Suppose that for every $t>0$,  $\Lambda_+ (t,.) $ and $\Lambda_- (t,.)$ are finite only on part of the real line; i.e.  there exist $ \mu^\ast_+ (t)$ and $ \mu^\ast_- (t)$ such that 
\[
0<p< \mu^\ast_\pm (t), ~~~\Lambda_\pm (t, p) < 0, ~~~~~~~Ê\textrm{and}~~~ \lim_{ p \to \mu^\ast_\pm (t)} \Lambda_\pm(t, p) = \infty.
\] 
Consider the local volatility function $\Sigma (t,x)$ such that the process $Y$ defined by the stochastic differential equation
\[
d Y_t =(q - \frac{1}{2}  \Sigma^2(t, Y_t)) d t +  \Sigma  (t, Y_t) d W_t ,
\]
generates the same marginal distributions as $X$. The following result links $ \Sigma(t,\pm y)$ for large $y$ to the behavior of $\Lambda_+(t,.)$  and $\Lambda_-(t,.)$ near the critical moments $\mu^\ast_+$ and $\mu^\ast_-$.

 \begin{theorem}\label{mgf_to_lv}
 If $\Lambda_+$ and $ \Lambda_-$ satisfy (ii)-(iii) in Assumption~\ref{hyp_laregdev}, then  for $y$ sufficiently large we have
\begin{equation}\label{localvol_for_mgf}
  \Sigma^2(t, \pm y ) = \sigma^\pm_0 (t) y +
   \frac{  \partial_t \Delta_\pm (t,\nu_\pm(y)) \mp q \mp   \sigma^\pm_0 (t) ( 1 \mp 2 \mu^\ast_\pm   \pm \frac{1}{\nu_\pm (y) }Ê) y /\nu_\pm(y)     }{  \frac{1}{2}  ( (\mu^\ast_\pm   - \frac{1}{\nu_\pm(y)} )^2 \mp (\mu^\ast_\pm   - \frac{1}{\nu_\pm(y)} )) \left(   1 +  (\gamma^2_\pm - \gamma_\pm )\frac{ \nu_\pm(y)^2}{2 y^2 \nu'_\pm (y) }   + o( \frac{ \nu_\pm(y)^2}{2 y^2 \nu'_\pm (y) })     \right) },
\end{equation}
where $ \gamma_\pm = \frac{\alpha_\pm}{\alpha_\pm + 1} $, $ \Delta_\pm (t, x) = \Lambda_\pm( t. \mu^\ast_\pm (t) - \frac{1}{x})$, with 
\[
 \sigma^\pm_0 (t) = \frac{ -2  \partial_t \mu^\ast_\pm (t)      }{ {\mu^\ast_\pm}^2 (t)  \mp  \mu^\ast_\pm  (t)    }
\]
and
\[
\nu_\pm (t, y) := (    \partial_\mu \Lambda_\pm (t, \mu^\ast_\pm (t) - \frac{1}{(.)}Ê) )^{-1} (y).
\]
 \end{theorem}

\begin{remark}
Note that the second term in the right hand side of (\ref{localvol_for_mgf}) is $ \mathcal{O} (y^\frac{-\alpha_\pm}{\alpha_\pm+1} )$. It is also worth noticing that $\frac{ \nu_\pm(y)^2}{2 y^2 \nu'_\pm (y) } \sim y^\frac{-\alpha_\pm}{\alpha_\pm +1}  $.
\end{remark}
\begin{proof} 
 Applying It\^o's formula to the process $ e^{ p Y_t}$ we find that, for $t>0$ and $p < \mu^\ast_+ (t)$,
\begin{equation}\label{first_edp}
\partial_t \psi^1_+ (t, p) =q \psi^1_+ (t, p) +  \frac{1}{2}   ( p^2 - p) \psi^{\Sigma^2(t,.)}_+ (t,p) ,
\end{equation}
where  $  \psi^1_+ (t, p)  :=  \mathbb{E}~( e^{ p Y_t} ) $, and for a function $f$  the notation $\psi^f(.)$ refers to
\[
\psi^f_+ (t,p) := \mathbb{E}~( f(Y_t) e^{ p Y_t} ).
\]
 Define $ \Delta_+ (t,x) := \ln \psi^1_+ (t, \mu^\ast_+ (t) - \frac{1}{x}))Ê$; we have 
 \[
 \partial_t    \Delta_+ (t,x) = \frac{   \partial_t  \psi^1_+ (t, \mu^\ast_+ (t) - \frac{1}{x})    }{   \psi^1_+ (t, \mu^\ast_+ (t) - \frac{1}{x})} 
 + \frac{   \partial_t \mu^\ast_+ (t) \partial_\mu   \psi^1_+ (t, \mu^\ast_+ (t) - \frac{1}{x})    }{   \psi^1_+ (t, \mu^\ast_+ (t) - \frac{1}{x})}.
 \]
 Noticing that $ \partial_x \Delta_+ (t,x) = \frac{    \partial_\mu   \psi^1_+ (t, \mu^\ast_+ (t) - \frac{1}{x}))    }{  x^2 \psi^1_+ (t, \mu^\ast_+ (t) - \frac{1}{x}))}  $ and dividing both sides of (\ref{first_edp}) by $ \psi^1_+ (t, \mu^\ast_+ (t) - \frac{1}{x}) $ we have
\begin{equation}\label{dt_Delta1}
\partial_t \Delta_+ (t,x)  =  q + \frac{1}{2}  ( (\mu^\ast_+ (t)  - \frac{1}{x} )^2 - (\mu^\ast_+ (t) - \frac{1}{x} ))  \frac{ \psi^{\Sigma(t,.)}_+ }{ \psi^1_+} (  t, \mu^\ast_+ (t) - \frac{1}{x} )+ \partial_t \mu^\ast_+ x^2 \partial_x \Delta_+(t,x).
\end{equation}
As highlighted in the introduction, we can actually derive an asymptotic formula for the local volatility by integrating and differentiating (\ref{expans_cdf}) with respect to $t$ and $x$. Without going through the formula, we can prove that $ \Sigma^2 (t,.)$ is smoothly varying  which allows us to apply Proposition~\ref{tauberian1}. This gives
 
\[
  \frac{ \psi^{\Sigma^2(t,.)} }{ \psi^1} (  t, \mu^\ast_+ - \frac{1}{x} ) = \Sigma(t,x^2 \partial_x \Delta (t,x) ) (  1 + \mathcal{O} (x^{-\alpha_+})).  
  \]
  It follows that
  \[
  \Sigma^2 (t, y       ) = \frac{ -2  \partial_t \mu^\ast_+ y + 2  \partial_t \Delta_+ (t, \nu (y) )     }{ {\mu^\ast_+}^2 (t)  - \mu^\ast_+  (t) + \frac{1}{\nu (y)} (1 - 2 \mu^\ast_+ + \frac{1}{\chi(y)} )   }(1 + \mathcal{O} (y^\frac{-\alpha_+}{\alpha_+ + 1})) ,
  \]
  where
  \[
  \nu(y) = [( .)^2 \partial_x \Delta_+ (t,.) ]^{-1} (y)~  \sim ~~ y^\frac{1}{\alpha_+ + 1} .
  \]
  Hence
  \[
  \Sigma^2 (t, y  ) = \frac{ -2  \partial_t \mu^\ast_+ (t)      }{ {\mu^\ast_+}^2 (t)  - \mu^\ast_+  (t)    } y + \mathcal{O} (y^\frac{\alpha}{\alpha + 1}).
  \]
    This means that 
    \[
    \Sigma^2(t,y ) = \frac{ -2  \partial_t \mu^\ast_+ (t)      }{ {\mu^\ast_+}^2 (t)  - \mu^\ast_+  (t)    } y + \Sigma_1 (t,y),
    \]
    where $ \Sigma_1 (t,y)= \mathcal{O} (y^\frac{\alpha}{\alpha + 1}) $. Rewriting the partial differential equation (\ref{dt_Delta1}) we have 
    \[
\partial_t \Delta_+ (t,x)  =  q + \frac{1}{2}  ( (\mu^\ast_+   - \frac{1}{x} )^2 - (\mu^\ast_+   - \frac{1}{x} )) (  \sigma_0 (t)  \frac{ \psi^{ Id}_+ }{ \psi^1_+}+  \frac{ \psi^{\Sigma_1(t,.)}_+ }{ \psi^1_+} (  t, \mu^\ast_+  - \frac{1}{x} ) )+ \partial_t \mu^\ast_+ x^2 \partial_x \Delta_+(t,x),
\]
where $ Id: z \mapsto z$ and 
\[
 \sigma_0 (t) = \frac{ -2  \partial_t \mu^\ast_+ (t)      }{ {\mu^\ast_+}^2 (t)  - \mu^\ast_+  (t)    }.
\]
Observe that, for any $\mu > 0$,
\[
\psi^{ Id}_+ (t, \mu)  :=   \mathbb{E} ( Y_t  e^{ \mu Y_t} ) = \partial_\mu \psi^{1}_+ (t, \mu).
\] 
In particular 
\[
\psi^{Id}_+   (  t, \mu^\ast_+  - \frac{1}{x} ) / \psi^{1}_+   (  t, \mu^\ast_+  - \frac{1}{x} )  = x^2 \partial_x \Delta_+ (t, x).
\]
It follows that
\[
 \partial_t \Delta_+ (t,x)  =  q +  \frac{1}{2 } (1 - 2 \mu^\ast_+  + \frac{1}{x}Ê) \sigma_0 (t) x \partial_x \Delta_+ (t,x)+
 \frac{1}{2}  ( (\mu^\ast_+   - \frac{1}{x} )^2 - (\mu^\ast_+   - \frac{1}{x} ))    \frac{ \psi^{\Sigma_1(t,.)}_+ }{ \psi^1_+} (  t, \mu^\ast_+  - \frac{1}{x} )  .
\]
Using Proposition~\ref{tauberian1}  once more we have
\[
\frac{ \psi^{\Sigma_1(t,.)}_+ }{ \psi^1_+} (  t, \mu^\ast_+  - \frac{1}{x} ) = \Sigma_1(t, x^2 \partial_x \Delta (t,x) ) (1 + (\gamma^2 - \gamma) \frac{ \Lambda'' ( \mu^\ast - \frac{1}{x}   )   }{ 2 {\Lambda'}^2 ( \mu^\ast - \frac{1}{x}   )  }    + o  (   \frac{ \Lambda'' ( \mu^\ast - \frac{1}{x}   )   }{ 2 {\Lambda'}^2 ( \mu^\ast - \frac{1}{x}   )  }       )),
\]  
where $ \gamma$ is the index of $ \Sigma_1 (t, .)$. It follows that
\begin{eqnarray*}
\partial_t \Delta_+ (t,\nu_+(y)) &=&  \sigma_0 (t) ( 1 - 2 \mu^\ast_+ + 1/\nu_+ (y)) y /\nu_+ (y) + \\ &&
\frac{1}{2}  ( (\mu^\ast_+   - \frac{1}{\nu_+(y)} )^2 - (\mu^\ast_+   - \frac{1}{\nu_+(y)} ))  \Sigma_1 (t,y) (1 +  (\gamma^2 - \gamma)\frac{ \nu_+(y)^2}{2 y^2 \nu'_+ (y) }    + o  (  \frac{ \nu_+(y)^2}{2 y^2 \nu'_+ (y) }  )      ).
\end{eqnarray*}
Hence $ \gamma =\frac{  \alpha_+}{\alpha_+ + 1} $ and we finally have
\[
 \Sigma_1 (t,y) = \frac{  \partial_t \Delta_+ (t,\nu_+(y)) - q-   \sigma^+_0 (t) ( 1 - 2 \mu^\ast_+   - \frac{1}{\nu_+ (y) }Ê) y /\nu_+(y)     }{  \frac{1}{2}  ( (\mu^\ast_+   - \frac{1}{\nu_+(y)} )^2 - (\mu^\ast_+   - \frac{1}{\nu_+(y)} )) \left(   1 +  (\gamma^2 - \gamma)\frac{ \nu_+(y)^2}{2 y^2 \nu'_+ (y) }  + o (  \frac{ \nu_+(y)^2}{2 y^2 \nu'_+ (y) }    )     \right) }  .
\]

  To prove the result for $ \Sigma (t, -k)$, we consider the process $ \tilde{Y}_t = - Y_t$. This process is given by the stochastic differential equation
  \[
d \tilde{Y}_t =(-q + \frac{1}{2}  \tilde{\Sigma}^2(t, \tilde{Y}_t)) d t + \tilde{\Sigma}  (t, \tilde{Y}_t) d \tilde{W}_t ,
\]
 where $\tilde{W}_t = - W_t  $ and  $\tilde{\Sigma}(t,y) =\Sigma (t, -y  ) $. Applying It\^o's formula to the process $ e^{ p \tilde{Y}_t}$ we find that
\[
\partial_t \psi^1_- (t, p) =-q \psi^1_- (t, p) +  \frac{1}{2}   ( p^2 +p) \psi^{\tilde{\Sigma}^2(t,.)}_- (t,p) ,
\]
 where $ \psi^f_- (t,p) := \mathbb{E} ( f(\tilde{Y}_t e^{ p \tilde{Y}_t} )  $. Define $ \Delta_- (t,x) := \ln \psi^1_- (t, \mu^\ast_- (t) - \frac{1}{x}))Ê$. We have
\[
\partial_t \Delta_- (t,x)  = -q+ \frac{1}{2}  ( (\mu^\ast_- (t)  - \frac{1}{x} )^2 + (\mu^\ast_- (t) - \frac{1}{x} ))  \frac{ \psi^{\tilde{\Sigma}^2(t,.)}_- }{ \psi^1_+} (  t, \mu^\ast_- (t) - \frac{1}{x} )+ \partial_t \mu^\ast_- x^2 \partial_x \Delta_-(t,x).
\]
We find in the same way as for $\Sigma(t,y)$ that
 \[
  \Sigma^2(t, - y ) = \sigma^-_0 (t) y +
   \frac{  \partial_t \Delta_- (t,\nu_-(y)) + q +   \sigma^-_0 (t) ( 1 + 2 \mu^\ast_-   - \frac{1}{\nu_- (y) }Ê) y /\nu_-(y)     }{  \frac{1}{2}  ( (\mu^\ast_-   - \frac{1}{\nu_-(y)} )^2 - (\mu^\ast_-   - \frac{1}{\nu_-(y)} )) \left(   1 +  (\gamma^2_- - \gamma_- )\frac{ \nu_-(y)^2}{2 y^2 \nu'_- (y) }   + o( \frac{ \nu_-(y)^2}{2 y^2 \nu'_- (y) })     \right) }
 \]
 with $\sigma^-_0$ and $\nu_-$ are given by the theorem.
  \end{proof}

\section{From moment explosion to  Implied volatility asymptotics }\label{section3}
\noindent Under Black Scholes model, the prices of standard European options are given explicitly via the famous Black-Scholes formula.  In particular the price of  an European Call option with strike $K$ and maturity $t$ is given by:
\[
 \mathbb{E} ( S_0 e^{    - \frac{\sigma^2}{2}Ê t + \sigma W_t}  - K)_+ = N(  \frac{1}{  \sigma \sqrt{t}  }  ( -k   + \frac{1}{2}   \sigma^2 t  )) -  e^{k }  N(  \frac{1}{ \sigma \sqrt{   t}}  ( -k   - \frac{1}{2}  \sigma^2 t   )) ,
\]
where $N$ denotes the cumulative distribution function of the standard Gaussian law, $ k = \ln (K/S_0)$ and  $ \sigma $ is the volatility parameter. For any price quote there exists a unique volatility parameter (that depends on $ t$ and $ k$) such that the Black-Scholes formula gives the same price: this is called the \textit{Black-Scholes implied volatility}.

Suppose that the market is described by a model $X$  (for ($ \log(S_t/S_0)   $)) under the  measure $\mathbb{P}$. Today's price of a Call option with maturity $t$ and  and strike $e^{k }$ is then given as function of the Black-Scholes implied volatility as follows:
\begin{equation}
 \mathbb{E} ( e^{ X_t} - e^k)_+   = N(  \frac{1}{ \sqrt{t} \sigma (t,k) }  ( -k + \frac{1}{2}  t  \sigma^2(t,k)   )) -  e^{k  }  N(  \frac{1}{   \sqrt{t}  \sigma(t,k)}  ( -k - \frac{1}{2} t \sigma^2 (t,k)  )).
\end{equation}
Differentiating both sides with respect to $k$ we get
\[
 -e^{k}~ \mathbb{P} (X_t \geq k) =  \partial_k  C_{BS}  ( k,  \sqrt{t} \sigma (t,k)) + \sqrt{t} \partial_k   \sigma (t,k)   C_{BS}  ( k, \sqrt{t}  \sigma (t,k)),
\]
where the notation $C_{BS}(k,v)$ refers to 
\[
C_{BS}(k, v) := N(  \frac{1}{ v }  ( -k + \frac{1}{2} v^2   )) -  e^{k  }  N(  \frac{1}{  v}  ( -k - \frac{1}{2} v^2  )).
\]
Hence
\[
 \mathbb{P} (X_t \geq k) = \tilde{N} (  \frac{1}{  \sqrt{t}  \sigma(t,k)}  ( k + \frac{1}{2} t  \sigma^2 (t,k)  )) -  \sqrt{t}  \partial_k \sigma  (t, k)    N' (  \frac{1}{ \sqrt{t}  \sigma(t,k)}  ( k + \frac{1}{2} t  \sigma^2 (t,k)   )) , 
\]
where $ \tilde{N} (x) =1- N(x)$. We clearly see that when $|k| $ is large, the quantity $ d(k) := \frac{1}{ \sqrt{t}  \sigma(t,k)}  ( k + \frac{1}{2} t  \sigma^2 (t,k)   ) \to \infty$.  The following bounds hold for any positive $y$ and obtained via an integration by part:
\[
 \frac{ N'(y) }{ y} (  1 - \frac{1}{ y^2} ) \le \tilde{N} (y) \le \frac{ N'(y) }{ y}.
\]
In particular, we have
\[
\tilde{N} (y) = \frac{N'(y)}{y} ( 1 + \mathcal{O} ( y^{-2} )).
\]
 Applying this to the black Scholes formula we get  
\begin{equation}\label{cumuldistrIV_r}
\mathbb{P} (X_t \geq k) =\frac{1}{ \sqrt{2 \pi}} e^{ - \frac{1}{2} d^2 (k) }Ê \left(  \frac{1}{d (k) } - \sqrt{t} \partial_k \sigma (t,k) + \mathcal{O} ( d(k)^{-3}    )     \right),  
\end{equation}
 Note that
\[
d^2 (k) = \frac{k^2}{ t  \sigma^2 (t,k)  } + \frac{1}{4} t \sigma^2 (t,k) + k.  
\]

The left wing is obtained in a very similar way; let $k$ be large enough and consider a European Put option with strike $e^{ -  k}$ and maturity $t$. Its price is given by
\begin{equation}
 \mathbb{E} ( e^{-k} - e^{ X_t}  )_+   = e^{-k} N( -\tilde{d} (k) ) -    N(-\tilde{d} (k) - \sqrt{t} \tilde{\sigma} (t, k)  ),
\end{equation}
where $\tilde{\sigma} (t,k) = \sigma (t, - k )$ and 
\[
\tilde{d} (k) =  \frac{1}{ \sqrt{t}  \tilde{\sigma} (t,k)}  ( k - \frac{1}{2} t  \tilde{\sigma} ^2 (t,k)   )
\]
Differentiating both sides with respect to $k$ we get 
\[
 \mathbb{P} (-X_t \geq k) = \tilde{N} ( \tilde{d} (t,k) ) -  \sqrt{t}  \partial_k  \tilde{\tilde{\sigma} }  (t,k)     N' ( \tilde{d} (t,k)   ) . 
\]
The bounds derived above hold also for  $  \tilde{d} (t,k) $, which gives 
\begin{equation}\label{cumuldistrIV_l}
\mathbb{P} (-X_t \geq k) =\frac{1}{ \sqrt{2 \pi}} e^{ - \frac{1}{2} \tilde{d}^2 (k) }Ê \left(  \frac{1}{\tilde{d} (x) } -\sqrt{t} \partial_k \tilde{\sigma}  (t,k) + \mathcal{O} ( \tilde{d}(k)^{-3}    )     \right).
\end{equation}

Using the bounds (\ref{cumuldistrIV_r}) and (\ref{cumuldistrIV_l}), the next theorem links the left/right wings of the Black-Scholes implied volatility to the moment generating function of $-X_t/X_t$.
\begin{theorem}\label{from_mgf_to_iv}
Suppose that for any $t > 0$,  $X_t$ and  $-X_t$ satisfy  Assumption~\ref{hyp_laregdev}. For $k$ sufficiently large we have
\begin{equation}
 t  \sigma^2 (t,\pm k)  = 4 \Lambda^\ast_\pm (t,k) +2 \tilde{c}^\pm_t (k) \mp 2k - 4 \sqrt{ ( \Lambda^\ast_\pm(t,k) +\frac{  \tilde{c}^\pm_t (k)}{2})  (  \Lambda^\ast_\pm(t,k) +\frac{  \tilde{c}^\pm_t (k)}{2} \mp k)  }    + \mathcal{O} (k^\frac{-\alpha_\pm \wedge 1}{\alpha_\pm+1} ) .
\end{equation}
with the notations $ \Lambda_\pm (t, \mu) := \ln \mathbb{E} e^{ \mu ( \pm X) } $,  $ \Lambda^\ast_\pm (t,.)$ is the Fenchel-Legendre transform of $\Lambda_\pm (t, .) $, $\alpha_\pm$ is the index of the regularly varying function $ x \mapsto \Lambda_\pm (t, \mu^\ast_\pm (t) - \frac{1}{x} )$ and $\tilde{c}_\pm$ is defined by
\begin{equation}
\tilde{c}^\pm_t (k) =    - \ln (  k \partial_{kk} \Lambda^\ast_\pm (t,k))+ 2 \ln\left(\sqrt{  \frac{ \mu^\ast_\pm (t) }{2} }Ê- \frac{\mu^\ast_\pm (t) }{2} \sqrt{  \mp2 + 4 \mu^\ast_\pm (t) - 4 \sqrt{ \mu^\ast_\pm ( \mu^\ast_\pm \mp 1) }   }  Ê\right).
\end{equation}
\end{theorem}

\begin{proof}
We have from Theorem~\ref{thm22_ext}
\[
\mathbb{P} (X_t \geq k) = e^{ - \Lambda^\ast_+ (t,k)}  \left( \frac{ \sqrt{ { \partial_x p^\ast_+} (t,k) } }{ p^\ast_+ (t,k) \sqrt{2 \pi}} - \frac{ 2 + \frac{\alpha_+}{ (\alpha_+ + 1)^2 } }{  24 \mu^\ast_+  \sqrt{2 \pi} } ~\frac{1}{k^2 \sqrt{ \partial_x {p^\ast_+}(t,x) }  }    + o(     \frac{1}{k^2  \sqrt{ \partial_x {p^\ast_+} (t,k) } }    )   \right).
\]
Comparing this with (\ref{cumuldistrIV_r}) we see that
\begin{equation}\label{first_equation}
  \frac{k^2}{ t  \sigma^2 (t,k)  } + \frac{1}{4} t \sigma^2 (t,k) + k -2 \ln (  \frac{1}{d (k) } - \sqrt{t} \partial_k \sigma   )  + \mathcal{O} (k^{-1} ) = 2 \Lambda^\ast_+ (t,k) - 2\ln( \frac{ \sqrt{ \partial_x { p^\ast_+}  } }{ p^\ast_+ (t,k)  }   ) + \mathcal{O} (k^\frac{-\alpha_+}{\alpha_++1} ) .
\end{equation}
Hence
\begin{eqnarray*}
 \frac{1}{ 4t  \sigma^2 (t,k)  }  &\times& \left(   t  \sigma^2 (t,k)  - 2 (   2 \Lambda^\ast_+ (t,k) - k - 2 \sqrt{  \Lambda^\ast_+(t,k) (  \Lambda^\ast_+(t,k) - k)  }  )      \right) \times
 \\ &&  \left(   t  \sigma^2 (t,k)  - 2 (   2 \Lambda^\ast_+ (t,k) - k + 2 \sqrt{  \Lambda^\ast_+(t,k) (  \Lambda^\ast_+(t,k) - k)  }  )      \right) = c^+_t (k) + \mathcal{O} (k^\frac{-\alpha_+}{\alpha_++1} ),
\end{eqnarray*}
where
\begin{equation}\label{def_c}
c^+_t (k) := -2 \ln( \frac{ \sqrt{ { p^\ast_+}' (t,k) } }{ p^\ast_+ (t,k)  }   ) +2 \ln (  \frac{1}{d (k) } - \sqrt{t} \partial_k \sigma (t,k) )  + \mathcal{O} (k^{-1} ) = \mathcal{O} ( \ln (k)).
\end{equation}
We emphasize that $ \Lambda^\ast_+ (t,k) = \mu^\ast_+ (t) k + \mathcal{O} (k^\frac{\alpha_+}{\alpha_++1})$. On the other hand, Lee's moment formula states that
\[
t  \sigma^2 (t,k) \sim k( -2 + 4 \mu^\ast_+ - 4 \sqrt{ \mu^\ast_+ ( \mu^\ast_+ - 1) }) .
\]
Noticing that 
\[
2 \left(   2 \Lambda^\ast_+ (t,k) - k - 2 \sqrt{  \Lambda^\ast_+(t,k) (  \Lambda^\ast_+(t,k) - k)  }  \right) \sim k \left( -2 + 4 \mu^\ast_+ - 4 \sqrt{ \mu^\ast_+ ( \mu^\ast_+ - 1) } \right) ,
\]
we conclude that
\[
 t  \sigma^2 (t,k)  = 2 \left(   2 \Lambda^\ast_+ (t,k) - k - 2 \sqrt{  \Lambda^\ast_+(t,k) (  \Lambda^\ast_+(t,k) - k)  } \right)  + \mathcal{O} ( \ln (k)).
\]
Plugging this value into the definition of $c^+_t(k)$, (\ref{def_c}), we find that $c^+_t(k) = \tilde{c}^+_t (k)  + \mathcal{O} ( k^\frac{-1}{\alpha + 1} )$, with
\[
\tilde{c}^+_t (k) =    - \ln (  k \partial_{kk} \Lambda^\ast_+ (t,k))+ 2 \ln\left(\sqrt{  \frac{ \mu^\ast_+ }{2} }Ê- \frac{\mu^\ast_+ }{2} \sqrt{  -2 + 4 \mu^\ast_+ - 4 \sqrt{ \mu^\ast_+ ( \mu^\ast_+ - 1) }   }  Ê\right) .
\]
Now if we go back to (\ref{first_equation}) and plug the value of $c^+_t (k) $  we have
\[
 \frac{k^2}{ t  \sigma^2 (t,k)  } + \frac{1}{4} t \sigma^2 (t,k) + k   = 2 \Lambda^\ast_+ (t,k) + \tilde{c}^+_t (k)  + \mathcal{O} (k^\frac{-\alpha_+ \wedge 1}{\alpha_++1} ) .
\]
We finally find that
\[
 t  \sigma^2 (t,k)  = 4 \Lambda^\ast_+ (t,k) +2\tilde{c}^+_t (k) - 2 k - 4 \sqrt{ ( \Lambda^\ast_+(t,k) +\frac{  \tilde{c}_+ (k)}{2})  (  \Lambda^\ast_+(t,k) +\frac{  \tilde{c}_+ (k)}{2} - k)  }    + \mathcal{O} (k^\frac{-\alpha \wedge 1}{\alpha+1} ) .
\]

The left wing is obtained in the same way: we write, using Theorem~\ref{thm22_ext},
\[
\mathbb{P} (-X_t \geq k) = e^{ - \Lambda^\ast_- (t,k)}  \left( \frac{ \sqrt{ \partial_x { p^\ast_-} (t,k) } }{ p^\ast_- (t,k) \sqrt{2 \pi}} - \frac{ 2 + \frac{\alpha_-}{ (\alpha_- + 1)^2 } }{  24 \mu^\ast_-  \sqrt{2 \pi} } ~\frac{1}{k^2 \sqrt{ \partial_x {p^\ast_-}(t,x) }  }    + o(     \frac{1}{k^2  \sqrt{ {p^\ast_-}' (t,k) } }    )   \right).
\]
We obtain, by comparing this with (\ref{cumuldistrIV_l}),
\[
 \frac{k^2}{ t  \tilde{\sigma}^2 (t,k)  } + \frac{1}{4} t \tilde{\sigma}^2 (t,k) - k -2 \ln (  \frac{1}{\tilde{d} (k) } + \sqrt{t} \partial_k \tilde{\sigma}  )  + \mathcal{O} (k^{-1} ) = 2 \Lambda^\ast_- (t,k) - 2\ln( \frac{ \sqrt{\partial_x  { p^\ast_-} } }{ p^\ast_- (t,k)  }   ) + \mathcal{O} (k^\frac{-\alpha_-}{\alpha_-+1} ) .
\]
We find in the same way as for the right wing that
\[
t  \tilde{\sigma}^2 (t,k)  = 4\Lambda^\ast_- (t,k) + 2 \tilde{c}^-_t (k) +2 k - 4 \sqrt{ ( \Lambda^\ast_-(t,k) +\frac{  \tilde{c}^-_t (k)}{2})  (  \Lambda^\ast_-(t,k) +\frac{ \tilde{c}^-_t (k)}{2} + k)  }    + \mathcal{O} (k^\frac{-\alpha_- \wedge 1}{\alpha_-+1} ) ,
\]
with
\[
\tilde{c}_-(t,k) =    - \ln (  k \partial_{kk} \Lambda^\ast_- (t,k))+ 2 \ln\left(\sqrt{  \frac{ \mu^\ast_- }{2} }Ê- \frac{\mu^\ast_- }{2} \sqrt{  2 + 4 \mu^\ast_- - 4 \sqrt{ \mu^\ast_- ( \mu^\ast_- + 1) }   }  Ê\right) .
\]
\end{proof}

\section{From SVI parametrization of implied volatility to moment explosion and local volatility}\label{section4}
\noindent  We remain in the same setup as the previous section, where $S$ denotes the stock price and $ X_t = \ln (S_t/S_0)$. We consider  Gatheral's SVI parametrization (cf. \cite{gatheral04}):
\begin{equation}\label{svi_vol}  
\sigma^2 (t,k) = a + b \left(     \rho ( k -m)  + \sqrt{ ( k -m)^2 + \eta^2 } \right).
\end{equation}
It is worth noticing that all the parameters depend on $t$. 
\begin{proposition}\label{mgf_SVI}
Suppose that the Black-Scholes implied volatility function generated from the option prices  
 is given as (\ref{svi_vol}). The moment generating functions of $ X_t $ and $(-X_t)$ explode at critical moment $\mu^\ast_+$ and $\mu^\ast_-$ respectively, where 
\begin{equation}
\mu^\ast_\pm (t) = \frac{1}{2} (   \frac{1}{b t (1 \pm \rho) } + \frac{1}{4} b t (1 \pm \rho) \pm 1 ).
\end{equation} 
Furthermore we have
\begin{equation}
\mathbb{E} ~e^{ ( \mu^\ast_\pm - \frac{1}{x} )  (\pm X_t) } = \mu^\ast_\pm    \xi_\pm e^{d^\pm_0 + \mu^\ast_\pm  m   } \left( ( 2 x) ^\frac{1}{2} +\mathcal{O} (1) \right), 
\end{equation}
where
\[
d^\pm_0 = \frac{1}{2} (-  m+  Ê\frac{  a}{ b^2 t (1 \pm \rho)^2}  \mp \frac{2 m  }{ b t (1 \pm \rho)} - \frac{a t}{4} ),
\]
and
\[
 \xi_\pm =  ( (2 \mu^\ast_\pm )^\frac{-1}{2} - \frac{1}{2} \sqrt{b t (1 \pm \rho) }Ê) .
\]
\end{proposition}

\begin{proof}
From (\ref{svi_vol}) we have, for $k$ sufficiently large,
\[
t \sigma^2(t,k) = a t + b t ( 1 + \rho) ( k - m) + \frac{b}{2} \frac{\eta^2}{ k -m} + \mathcal{O}  (k-m)^{-3}  ,
\]
It follows that $d^2(y+m)$ may be written as   
\begin{equation}\label{d2_ym}  
\frac{-1}{2}Êd^2(y+m) =   - \frac{1}{2} (\frac{1}{b t (1 + \rho) } + \frac{1}{4} bt (1 + \rho) + 1) Êy + d^+_0 + \mathcal{O} (y^{-1}), 
\end{equation}
where
\begin{eqnarray*}
d^+_0 &=&Ê\frac{1}{2} (  -m+  Ê\frac{  a}{ b^2 t(1 + \rho)^2}  -\frac{2 m  }{ b t(1 + \rho)} - \frac{at}{4} )  ,
\end{eqnarray*}
Note also that
\[
\frac{1}{d(y+m)} - \sqrt{t}  \partial_k \sigma(y+m) = ( (2 \mu^\ast)^\frac{-1}{2} - \frac{1}{2} \sqrt{b t (1 + \rho) }Ê) y^\frac{-1}{2} + (\frac{d^+_0 }{2} (2 \mu^\ast)^\frac{-3}{2} + \frac{a}{ 4 \sqrt{b t (1 + \rho) }}  )  Êy^\frac{-3}{2} + \mathcal{O} (y^{-5/2}),
\]
where   
\[
{ \mu^\ast_+}  = \frac{1}{2} (   \frac{1}{b  t (1 + \rho) } + \frac{1}{4} b t (1 + \rho) + 1 ) .
\] 
We may differentiate with respect to $k$ and consider the density function and then express the moment generating function in terms of the density. But to stay in general setting (i.e. without assuming the existence of a density function for $X_t$) we first use following representation of the exponential function as in (\ref{reprentation_exponential})
\[
\forall c, p, k, X \in R, ~~~~~~e^{p X} =  (1 \wedge c)~ e^{ p (X \wedge c) } +p  \int_c^\infty  e^{ p z} 1_{ X  \geq z} d z.
\]
It follows that for any $p >0$ such that $ \mathbb{E} e^{ p X_t} < \infty$ and for any $ c > 0$ we have
\[
 \mathbb{E} e^{ p X_t} =  (1 \wedge c)  \mathbb{E}  ~ e^{ p (X \wedge c) }+ p \int_c^\infty e^{ p z} \mathbb{P} (X_t \geq z) d z.
\]
It is easy to see that $ \mathbb{E} e^{ \mu^\ast X_t} = \infty$. For $x$ sufficiently large, we have
\[
 \mathbb{E} ~e^{ ( \mu^\ast_+ - \frac{1}{x} )  X_t} = (1 \wedge c) \mathbb{E}  e^{ ( \mu^\ast_+ - \frac{1}{x} ) c\wedge X_t} + \frac{1}{ \sqrt{ 2 \pi}}Ê ( \mu^\ast_+ - \frac{1}{x} ) 
e^{ ( \mu^\ast_+ - \frac{1}{x} ) m   } \int_{c-m}^\infty e^{ - \frac{1}{x}  z + d^+_0 +  \mathcal{O} ( z^{ -1} ) } ( \xi_+ z^\frac{-1}{2} + \mathcal{O}( z^\frac{-3}{2})   ) d z,
\]
where $ \xi_+=  ( (2 \mu^\ast_+)^\frac{-1}{2} - \frac{1}{2} \sqrt{b (1 + \rho) }Ê) $. Choosing $c= 1 + m$, we find that 
\[
\mathbb{E} ~e^{ ( \mu^\ast_+ - \frac{1}{x} )  X_t } =    (1 \wedge c) \mathbb{E}  e^{ ( \mu^\ast_+ - \frac{1}{x} ) c\wedge X_t}  + \frac{1}{ \sqrt{ 2 \pi}} x ( \mu^\ast_+ - \frac{1}{x} ) e^{d^+_0 +( \mu^\ast_+ - \frac{1}{x} ) m   } 
 \int_{\frac{1}{x} }^\infty e^{ - z} (  \xi_+ ( z x)^\frac{-1}{2} + \mathcal{O} ( (zx)^{-1} )  ) d z.
\]
Hence
\[
 \mathbb{E} ~e^{ ( \mu^\ast_+ - \frac{1}{x} )  X_t} =    (1 \wedge c) \mathbb{E}  e^{ ( \mu^\ast_+ - \frac{1}{x} ) c\wedge X_t} + x ( \mu^\ast_+ - \frac{1}{x} ) e^{d^+_0 +( \mu^\ast_+ - \frac{1}{x} ) m   } ( \xi_+ (  \frac{1 }{2}  x) ^\frac{-1}{2}  + \mathcal{O} (x^{ -1})    ).
\]
Observing that $\mathbb{E}  e^{ ( \mu^\ast_+ - \frac{1}{x} ) c\wedge X_t} \le e^{ \mu^\ast_+ |c|}$, we finally have
\[
  \mathbb{E} ~e^{ ( \mu^\ast_+ - \frac{1}{x} )  X_t} = \mu^\ast_+    \nu_ 1 e^{d^+_0 + \mu^\ast_+  m   } (  2  x) ^\frac{1}{2} +\mathcal{O} (1).
\]
  
  For $(-X_t)$ we proceed in a very similar way; we write, for $k$ sufficiently large,
  \begin{equation}\label{cumuldistrIV}
\mathbb{P} (-X_t \geq k) =\frac{1}{ \sqrt{2 \pi}} e^{ - \frac{1}{2} \tilde{d}^2 (k) }Ê \left(  \frac{1}{\tilde{d} (x) } -\sqrt{t} \partial_k \tilde{\sigma} (t,k) + \mathcal{O} ( \tilde{d}(k)^{-3}    )     \right),  
\end{equation}
 where
 \[
\tilde{d} (k) =  \frac{1}{ \sqrt{t}  \tilde{\sigma}(t,k)}  ( k - \frac{1}{2} t  \tilde{\sigma}^2 (t,k)   ),
\]
and
\[
\tilde{\sigma}^2 (t,k) := \sigma^2 (t,-k) = a + b \left(    -\rho ( k +m)  + \sqrt{ ( k +m)^2 + \eta^2 } \right)
\]
We find in the same way as before that
\begin{equation}\label{d2_ym} 
\frac{-1}{2}Ê\tilde{d}^2(y- m) =   - \frac{1}{2} (\frac{1}{b t (1 - \rho) } + \frac{1}{4} bt (1 - \rho) - 1) Êy + d^-_0   + \mathcal{O} (y^{-1}), 
\end{equation}
where
\begin{eqnarray*}
d^-_0 &=&Ê\frac{1}{2} (-  m+  Ê\frac{  a}{ b^2 t(1 - \rho)^2}  +\frac{2 m  }{ b t(1 - \rho)} - \frac{at}{4} )  ,
\end{eqnarray*}
Note also that
\[
\frac{1}{ \tilde{d}(y-m)}  \sqrt{t}  \partial_k \tilde{\sigma}(t, y-m) = ( (2 \mu^\ast_-)^\frac{-1}{2} - \frac{1}{2} \sqrt{b t (1 - \rho) }Ê) y^\frac{-1}{2} +   \mathcal{O} (y^{-3/2}),
\]
where   
\[
{ \mu^\ast}_-  = \frac{1}{2} (   \frac{1}{b  t (1 - \rho) } + \frac{1}{4} b t (1 - \rho) - 1 ) .
\] 
The rest is just repeating what we have done for $X_t$.
\end{proof}

We draw attention to the fact that a sharp asymptotic expansion of the price of standard European options (namely Call and Put options) is easily derived under the SVI parametrization of the implied volatility surface. It is given by the following result:
\begin{proposition}
The following expansions holds for European Call and Put options with strike $e^{k}$ and $e^{-k}$ respectively:
\begin{equation}\label{call_SVI}
 \mathbb{E} (e^{ X_t} - e^k )_+= \frac{1}{ \sqrt{2 \pi}} \frac{  \sqrt{  b t (1 + \rho)} }{2 \mu^\ast_+ - \sqrt{  2 \mu^\ast_+ b t (1 + \rho)} } 
 k^\frac{-1}{2}
 e^{-( \mu^\ast_+ -1)Êk  + d^+_0+ m\mu^\ast_+}  (  1 +  \mathcal{O} (k^{-1})) ,
\end{equation}
and
\begin{equation}\label{put_SVI}
  \mathbb{E} (e^{ -k}-e^{ X_t}  )_+= \frac{1}{ \sqrt{2 \pi}} \frac{  \sqrt{ bt (1 - \rho) } }{2 \mu^\ast_- - \sqrt{   \mu^\ast_- bt (1 - \rho) } } k^\frac{-1}{2} e^{-( \mu^\ast_- + 1) Êk  + d^-_0+ m\mu^\ast_-  }  (  1 +  \mathcal{O} (k^{-1})) ,
\end{equation}
with the same notations as in Proposition~\ref{mgf_SVI}.
\end{proposition}

\begin{proof}
We have 
\begin{eqnarray*}
C(t,k) &=& \mathbb{E} (e^{ X_t} - e^k )_+= e^{k} \tilde{N} (d(k)) - \tilde{N} ( d (k) - \sqrt{t} \sigma (t,k)) 
\\ &=&  
   \frac{ N' (d(k) -\sqrt{t} \sigma (t,k) }{ d(k) -\sqrt{t} \sigma (t,k)} (1 + \mathcal{O} (( d(k) -\sqrt{t} \sigma (t,k))^2) ) - e^{k} \frac{ N' (d(k)) }{ d(k)} (1 + \mathcal{O} (d(k)^2) ) 
\\ &=&  
  \frac{1}{ \sqrt{2 \pi}}  e^{k  - \mu^\ast_+ Ê(k-m) + d^+_0 + \mathcal{O} (k^{-1})  }\left( \frac{1 }{ d(k) -\sqrt{t} \sigma (t,k)} - \frac{1 }{ d(k)}  +  \mathcal{O}  (k^\frac{-3}{2})  \right)
\\ &=&  
\frac{1}{ \sqrt{2 \pi}} \frac{  \sqrt{ bt(1 + \rho)} }{2 \mu^\ast_+ - \sqrt{ 2  \mu^\ast_+ bt(1 + \rho)} } e^{k  - \mu^\ast_+ Ê(k-m) + d^+_0} k^\frac{-1}{2} (  1 +  \mathcal{O} (k^{-1})) .
\end{eqnarray*}
The same analysis applies to Put option with small strike; we write
\begin{eqnarray*}
P(t,k) &=& \mathbb{E} (e^{ -k}-e^{ X_t}  )_+ = e^{-k} \tilde{N} (\tilde{d}(k)) - \tilde{N} ( \tilde{d} (k) + \sqrt{t} \sigma (t,-k)) 
\\ &=&   
  \frac{1}{ \sqrt{2 \pi}}  e^{-k  - \mu^\ast_- Ê(k-m) + d^-_0 + \mathcal{O} (k^{-1})  }\left(\frac{1 }{ \tilde{d}(k)}- \frac{1 }{ \tilde{d}(k) +\sqrt{t} \sigma (t,-k)} -   +  \mathcal{O}  (k^\frac{-3}{2})  \right)
\\ &=&  
\frac{1}{ \sqrt{2 \pi}} \frac{  \sqrt{ bt (1 - \rho) } }{2 \mu^\ast_- - \sqrt{   \mu^\ast_+ bt (1 - \rho) } } e^{-k  - \mu^\ast_- Ê(k-m) + d^-_0} k^\frac{-1}{2} (  1 +  \mathcal{O} (k^{-1})) .
\end{eqnarray*}
\end{proof}

From this expansion we can derive the large asymptotics of the local volatility function. We emphasize that Dupire's local volatility is defined as
\begin{equation}\label{lv_positive}
\Sigma^2 (t, k) = \frac{2  \partial_t \mathbb{E} (e^{X_t} - e^k)_+ }{ (\partial_{kk} - \partial_k)\mathbb{E} (e^{X_t} - e^k)_+ } ~Ê= ~ \frac{ 2 \partial_t \mathbb{E} (e^{k}-e^{X_t} )_+ }{ (\partial_{kk} - \partial_k)\mathbb{E} (e^{k}-e^{X_t} )_+ } .
\end{equation}
and
\begin{equation}\label{lv_negative}
\Sigma^2 (t,- k) = \frac{2  \partial_t \mathbb{E} (e^{X_t} - e^{-k})_+ }{ (\partial_{kk} + \partial_k)\mathbb{E} (e^{X_t} - e^{-k} )_+ } ~Ê= ~ \frac{ 2 \partial_t \mathbb{E} (e^{-k}-e^{X_t} )_+ }{ (\partial_{kk} + \partial_k)\mathbb{E} (e^{-k}-e^{X_t} )_+ } .
\end{equation}

\begin{corollary}
Suppose that all prices of European Call options can be recovered from the SVI parametrization (\ref{svi_vol}). Then the local volatility function is given by
\[
\Sigma^2 (t, \pm y) =    \frac{ -2  \partial_t \mu^\ast_\pm (t)      }{ {\mu^\ast_\pm}^2 (t)  \mp \mu^\ast_\pm  (t)    } (y -  \frac{ 2 \mu^\ast_\pm \mp 1   }{2 ( {\mu^\ast_\pm}^2    \mp \mu^\ast_\pm  )  } ) )  +    \partial_t c_\pm (t)         + \mathcal{O} (y^\frac{-1}{2} ) ,
\]
where $ c_\pm (t)  = \ln (  \frac{1}{ \sqrt{2 \pi}} \frac{  \sqrt{  b t (1 \pm \rho)} }{2 \mu^\ast_\pm - \sqrt{  2 \mu^\ast_\pm b t (1 \pm \rho)} }    )  + d^\pm_0+ m\mu^\ast_\pm $, with the same notations as in Proposition~\ref{mgf_SVI}.
\end{corollary}

\begin{proof}
It follows from (\ref{call_SVI}), (\ref{call_SVI}), (\ref{lv_positive}) and (\ref{lv_negative})
 \end{proof}

\begin{remark}
This result could have been proved differently using Theorem~\ref{mgf_to_lv} and Proposition~\ref{mgf_SVI}.
\end{remark}

\section{  Heston's model and Stein-Stein model}\label{section5}
\noindent In this section we consider the application of the results obtained in the previous sections to two examples of stochastic volatility model: Heston's model and Stein-Stein model

\subsection{Heston's model}
 The Heston model is defined by the stochastic differential equation:
\begin{eqnarray}
d (\ln (S_t)) &=& - \frac{V_t}{ d t} + \sqrt{V_t} d B_t , \nonumber\\ 
d V_t &=&(a - b V_t ) d t + \sigma \sqrt{V_t} d W_t ,\nonumber\\
d \langle W, B \rangle_t  &=& \rho dt.
\end{eqnarray}
This model is known to have the moment explosion property. In the case of negative correlation, Friz and Gerhold \cite{friz15}   derive an equivalence for the local volatility corresponding to the Heston model  as 
\[
\Sigma^2 (t, y) \sim  \frac{ 2}{ \mu^\ast_+ (t) ( \mu^\ast_1 - 1) R_1 (t) R_2 (t)}  y, ~~~~~~~\textrm{as}~~~k \to \infty.
\]
with an explicit expression for $R_1$ and $R_2$.  It is clear that Theorem~\ref{mgf_to_lv} gives much more accurate expansion for the local volatility. It also applies to all cases (:$ k \to - \infty$ and $ \rho > 0$).

In order to apply Theorem~\ref{mgf_to_lv} to the Heston model, we need first to calculate the moment generating function of $ X_t = \ln (S_t/S_0)$. For $\mu > 0$ we have
\begin{eqnarray*}
\mathbb{E} \; e^{\mu X_t}  &=&  \mathbb{E} \; e^{ - \frac{\mu}{2} \int_0^t V_s d s + \mu \rho \int_0^t \sqrt{V_s} d W_s + \mu \sqrt{1 - \rho^2} \int_0^t \sqrt{V_s} d W^2_s } \nonumber\\ &=&
 \mathbb{E} \left\{  e^{  \frac{\mu^2 (1 - \rho^2) -\mu}{2} \int_0^t V_s d s + \mu \rho \int_0^t \sqrt{V_s} d W_s   }
 \left[  \mathbb{E} \; e^{   \mu \sqrt{1 - \rho^2} \int_0^t \sqrt{V_s} d W^2_s - \frac{\mu^2 (1 - \rho^2)}{2} \int_0^t V_s d s} \Big| (W_s)_{s \le t} \right]  \right\} \nonumber\\ &=&
 \mathbb{E} \left[  e^{ \mu \rho \int_0^t \sqrt{V_s} d W_s - \frac{\mu^2 \rho^2}{2} \int_0^t V_s d s } \; e^{  \frac{\mu^2 -\mu}{2} \int_0^t V_s d s      }  \right] \nonumber\\ &=& 
  \mathbb{E}^\mathbb{Q} \;    e^{  \frac{\mu^2 -\mu}{2} \int_0^t V_s d s      }   ,
\end{eqnarray*}
where we used the law of iterated conditional expectation and the fact that $V_s$ is measurable with respect to $W$. The last inequality is a consequence of Girsanov theorem, where under $ \mathbb{Q}$, the process $V$ satisfies the stochastic differential equation
\begin{equation}
d V_t = \left(  a - ( b - \rho \sigma \mu )V_t\right)d t + \sigma \sqrt{V_t} d W^{\mathbb{Q}}_t
\end{equation}
with $ \mathbb{Q}-$Brownian motion $W^{1,\mathbb{Q}} $. Calculating $\mathbb{E} \; e^{\mu X_t}$ is then reduced to the calculation of the moment generating function of the time average of the CIR process $V$ under $\mathbb{Q}$; this is given explicitly in \cite{Aly13} by
\[
\mathbb{E}^\mathbb{Q} ~e^{ \frac{\mu^2 -\mu}{2}  \int_0^t V_s d s   } = e^{ a \varphi(t; \mu) + v_0 \psi(t; \mu)},
\]
where $ \varphi(t; \mu) = \int_0^t \psi (s; \mu) d s$ and
\begin{equation}\label{psi_CIR}
\psi(t) = \left\{
\begin{array}{ll} \psi_1 (t; \mu ) := 
\frac{b - \rho \sigma \mu }{\sigma^2} -  \frac{\sqrt{ c_1 ( \mu)  } }{\sigma^2 }  \frac{   c_2 ( \mu ) e^{ \sqrt{c_1 (\mu)   } t}    +1 }{   c_2 (\mu) e^{ \sqrt{c_1 (\mu)   } t}    - 1     }, 
&  if ~~  0 \le  \frac{\mu^2 - \mu }{2}  \le \frac{ (b - \rho \sigma \mu)^2 }{ 2 \sigma^2},   \\
&
\\ \psi_2 (t; \mu) :=
\frac{ b - \rho \sigma \mu }{\sigma^2} +  \frac{\sqrt{ -c_1(\mu) }}{\sigma^2 } 
\tan \left(
\sqrt{  -c_1 (\mu)} \frac{t}{2}Ê+ \arctan(  \frac{   -b + \rho \sigma \mu  }{  \sqrt{ - c_1 (\mu) }      }      )
\right),
&   if ~~~  \frac{\mu^2 - \mu }{2}  > \frac{ (b - \rho \sigma \mu)^2 }{ 2 \sigma^2}  ,
\end{array}
\right.
\end{equation}
where $c_1$ and $c_2$ are defined by 
\begin{equation}\label{definition_c_1_2}
c_1 (\mu ) = (b - \rho \sigma \mu )^2 -  \sigma^2 (\mu^2 - \mu), ~~~~~~~\textrm{and} ~~~~~~
c_2 (\mu ) = \frac{ b - \rho \sigma \mu + \sqrt{ c_1 (\mu) }}{ b - \rho \sigma \mu - \sqrt{ c_1 (\mu) }    } .
\end{equation}
 The function $ \varphi (.)$ is given by
\begin{equation}\label{phi_CIR}
\varphi(t) = \left\{
\begin{array}{ll} \varphi_1 (t; \mu) := 
\frac{ b - \rho \sigma \mu + \sqrt{c_1 (\mu)   }}{\sigma^2} t -  \frac{2}{\sigma^2 }  
\ln \left(
\frac{  c_2 (\mu)  e^{  \sqrt{c_1 (\mu)   }   t}    - 1}{  c_2 (\mu)      -1    }
\right), 
&   if ~~~   0 \le  \frac{\mu^2 - \mu }{2}  \le \frac{ (b - \rho \sigma \mu)^2 }{ 2 \sigma^2},   \\
&
\\  \varphi_2 (t; \mu) :=
\frac{\tilde{b}}{\sigma^2} t  -  \frac{2 }{\sigma^2 } 
\ln   \left(  \frac{
\cos \left( 
\sqrt{ -c_1 (\mu)} \frac{t}{2}Ê+ \arctan(  \frac{   - b + \rho \sigma \mu  }{  \sqrt{  - c_1 (\mu) }      }      )
\right)
}{ \cos \left(
  \arctan(  \frac{   - b + \rho \sigma \mu  }{  \sqrt{  - c_1 (\mu) }      }       )
  \right)
  }
\right)  ,
&   if ~~~  \frac{\mu^2 - \mu }{2}  > \frac{ (b - \rho \sigma \mu)^2 }{ 2 \sigma^2}  .
\end{array}
\right.
\end{equation}
 
 \begin{remark}
 We deliberately did not consider the case $ \frac{\mu^2 - \mu }{2} < 0$. The reason is that when $ \frac{\mu^2 - \mu }{2} < 0$, the $\psi(t, \mu)$ might take a third form depending on other parameters (see Appendix~A in \cite{Aly13}). This, however, will not affect our analysis as $ \frac{\mu^2 - \mu }{2} < 0$ happens only when $ 0 < \mu \le 1$. As it is well known that $S$ is a true martingale,  the "positive" critical moment of $X$ is always larger than 1.
  \end{remark}
 
From  (\ref{psi_CIR}) and (\ref{phi_CIR})  we clearly see that  there exists $ \mu^\ast_+ (t) > 1$   (resp $\mu^\ast_- (t)>0$) such that the moment generating function of $X_t$ (resp. $-X_t$) is finite between 0 and $\mu^\ast_+ (t) $ (resp. $\mu^\ast_- (t)$) and explodes at $\mu^\ast_+ (t) > 0$ (resp. $\mu^\ast_- (t)$). The next result is nothing but rewriting the moment generating function of $X_t$ and $-X_t$ in terms of (\ref{psi_CIR}) and (\ref{phi_CIR}). 

\begin{proposition}
Define  $ \hat{\mu}_+$ and $ \hat{\mu}_-$ by
\begin{equation}
 \hat{\mu}_+ =  \min \left\{  \mu > 0: ~~~~c_1 (\mu)= 0  \right\} = \frac{  \sigma^2 - 2 b \rho \sigma  + \sqrt{ (\sigma^2 - 2 b \rho \sigma)^2  + 4 b^2 \sigma^2 ( 1 - \rho^2)     }   }{ 2 \sigma^2 ( 1 - \rho^2)} ,
\end{equation}
 and 
\begin{equation}
 \hat{\mu}_- =  \min \left\{  \mu > 0: ~~~~c_1 (-\mu)= 0  \right\} = \frac{ -( \sigma^2 - 2 b \rho \sigma)  + \sqrt{ (\sigma^2 - 2 b \rho \sigma)^2  + 4 b^2 \sigma^2 ( 1 - \rho^2)     }   }{ 2 \sigma^2 ( 1 - \rho^2)} .
\end{equation}
where $c_1 $ is defined by (\ref{definition_c_1_2}). For any $\mu > 0$ we have
\begin{equation}
\mathbb{E} \; e^{\mu X_t} 1_{ \mu \geq 1} = e^{a \varphi_1 (t; \mu) + v_0 \psi_1(t; \mu)} 1_{ 1 \le \mu \le  \hat{\mu}_+ \wedge \mu^\ast_+ (t)} +
e^{a \varphi_2 (t; \mu) + v_0 \psi_2(t; \mu)} 1_{ \hat{\mu}_+ < \mu <  \mu^\ast_+ (t)} ,
\end{equation}
and
\begin{equation}
\mathbb{E} \; e^{-\mu X_t}  1_{ \mu > 0} = e^{a \varphi_1 (t; -\mu) + v_0 \psi_1(t; -\mu)} 1_{0< \mu <  \hat{\mu}_- \wedge \mu^\ast_- (t)} +
e^{a \varphi_2 (t; -\mu) + v_0 \psi_2(t; -\mu)} 1_{ \hat{\mu}_- < \mu <  \mu^\ast_- (t)} ,
\end{equation}
where $ \varphi_{1,2}$ and $ \psi_{1,2}$ are given in (\ref{psi_CIR}) and (\ref{phi_CIR})
\end{proposition}

\begin{proof}
This follows immediately from (\ref{psi_CIR}) and (\ref{phi_CIR}) and the fact that for $\mu> 1$ or $ \mu < 0$, we have $0 \le  \frac{\mu^2 - \mu }{2} $ and 
\[
    \frac{\mu^2 - \mu }{2}  < \frac{ (b - \rho \sigma \mu)^2 }{ 2 \sigma^2} ~~\Longleftrightarrow ~~ c_1 (\mu) > 0.
\]
\end{proof}

  The moment generating function of $X$ and $-X$ near the critical moments $ \mu^\ast_+$ and $\mu^\ast_-$ are given in the following result.
  \begin{proposition}\label{expans_lambda_heston}
   Denote $ \Lambda_\pm (t, \mu) \ln \mathbb{ E} e^{ \pm \mu X_t}$. For $\mu $ sufficiently close to $\mu^\ast_+ (t)$ (resp. $\mu^\ast_- (t)$) $\Lambda_\pm (t, \mu)$ is given as
   \begin{equation}\label{expansion_mgf_heston}
   \Lambda_\pm (t, \mu) = \frac{   \omega_\pm (t) }{  \mu^\ast_\pm (t) - \mu   } + \frac{2 a}{\sigma^2 }  \ln (    \frac{1}{ \mu^\ast - \mu}  ) + m_\pm (t) + \sum_{ i \geq 1} d^\pm_i (t) ( \mu^\ast_\pm (t)  - \mu)^i
   \end{equation}
     \end{proposition}

\begin{proof}
See Theorem~3.2 in \cite{Aly13}. A similar result is found in \cite{friz11}.
\end{proof}

 \begin{remark}
It is difficult, or even impossible to obtain an explicit expression for $\mu^\ast_+$ and $\mu^\ast_-$. However, all other coefficients of (\ref{expansion_mgf_heston}) may be obtained  explicitly (or numerically) in terms of the model parameters and $ \mu^\ast_\pm$; to calculate numerically, one needs just to see them as
\begin{eqnarray*}
\omega_\pm (t) &=& \lim_{\mu \to \mu^\ast_\pm (t) } ( \mu^\ast_\pm (t) - \mu) \left( \Lambda_\pm (t, \mu) - \frac{2 a}{\sigma^2 }  \ln (    \frac{1}{ \mu^\ast - \mu}  )  \right) ,
\\ 
m_\pm  (t) &=& \lim_{\mu \to \mu^\ast_\pm (t) }   \left( \Lambda_\pm (t, \mu) - \frac{   \omega_\pm (t) }{  \mu^\ast_\pm (t) - \mu   } - \frac{2 a}{\sigma^2 }  \ln (    \frac{1}{ \mu^\ast - \mu}  )  \right), ~~~~~~~~~~ \textrm{and for }~i \geq 1,
\\
d^\pm_i (t) &=&\lim_{\mu \to \mu^\ast_\pm (t) }   \frac{1}{( \mu^\ast - \mu)^i}   \left( \Lambda_\pm (t, \mu) - \frac{   \omega_\pm (t) }{  \mu^\ast_\pm (t) - \mu   } - \frac{2 a}{\sigma^2 }  \ln (    \frac{1}{ \mu^\ast - \mu}  ) - m_\pm (t)  \right)
\end{eqnarray*}
  \end{remark}

 We clearly see that the Heston model satisfies the assumption of Theorem~\ref{mgf_to_lv}. The next result gives a sharp asymptotic expansion of the local volatility under Heston's model.
 \begin{proposition}\label{local_vol_heston}
 The local volatility corresponding to the Heston model satisfies
 \begin{equation}\label{localvol_heston}
  \Sigma^2(t, \pm y ) = \sigma^\pm_0 (t) y +
   \frac{   \omega'_\pm (t) \tilde{\nu}_\pm(y) + m'_\pm(t) \mp q \mp   \sigma^\pm_0 (t) ( 1 \mp 2 \mu^\ast_\pm   \pm \frac{1}{\tilde{\nu}_\pm (y) }Ê) y / \tilde{\nu}_\pm(y)     }{  \frac{1}{2}  ( (\mu^\ast_\pm   - \frac{1}{\tilde{\nu}_\pm(y)} )^2 \mp (\mu^\ast_\pm   - \frac{1}{\tilde{\nu}_\pm(y)} )) \left(   1 - \frac{1}{4} \frac{ \tilde{\nu}_\pm(y)^2}{2 y^2 {\tilde{\nu}_\pm}' (y) }       \right) } + o(1),
\end{equation}
where
\[
 \sigma^\pm_0 (t) = \frac{ -2  \partial_t \mu^\ast_\pm (t)      }{ {\mu^\ast_\pm}^2 (t)  \mp  \mu^\ast_\pm  (t)    },
\]
and
\begin{equation}\label{approx_nu}
\tilde{\nu}_\pm (t, y) = \frac{ 1}{ \sqrt{ \omega_\pm (t) } } (y+  \frac{a^2}{\sigma^4 \omega_\pm (t) } )^\frac{1}{2}    - \frac{a}{\sigma^2 \omega_\pm (t) }  .
\end{equation}
 \end{proposition}
 
   \begin{proof}
  The first statement of the proposition follows from a direct application of Theorem~\ref{mgf_to_lv}.  For the approximation (\ref{approx_nu}) we observe that $\nu_\pm (t,y)$ is the unique solution of
  \[
  \omega_\pm (t) x^2 + \frac{2 a}{ \sigma} x - d^\pm_1 - \sum_{ i \geq 1} (i+1) d^\pm_{i+1} (t) x^{ - i  } = y.
  \]
  This can be approximated by taking the unique positive solution to
  \[
   \omega_\pm (t) x^2 + \frac{2 a}{ \sigma} x   = y : ~~~~~~ x = \tilde{\nu}_\pm (y) =  \frac{ 1}{ \sqrt{ \omega_\pm (t) } } (y+  \frac{a^2}{\sigma^4 \omega_\pm (t) } )^\frac{1}{2}    - \frac{a}{\sigma^2 \omega_\pm (t) } .
  \]
  It is easy then to see that the error of this approximation of order $ y^\frac{-1}{2}$; that is
  \[
  \nu_\pm (t,y) = \tilde{\nu}_\pm (y) + \mathcal{O} (y^\frac{-1}{2} ).
  \]
  Choosing $\nu_\pm$ or $\tilde{\nu}_\pm$ is then equivalent in the formula since the error is $ o(1)$.
  \end{proof}
 
 \begin{remark}
 The error in the expansion (\ref{localvol_heston}) is actually $\mathcal{O}(y^\frac{-1}{4}) $ which is much more precise than $o(1)$ (the difference between this approximation and the actual local volatility tends to $0$ as $y \to \infty$ anyway). Indeed,  the term $o(1)$ comes from the product  $o( \frac{ \tilde{\nu}_\pm(y)^2}{2 y^2 {\tilde{\nu}_\pm}' (y) }  ) $, which is equivalent to $ o(y^\frac{-1}{2})$ and the numerator of the second term in the right hand side of (\ref{localvol_heston}),  which is  $\approx ~  y^\frac{1}{2}$.   We can see by looking at Theorem~\ref{thm22_ext} and its proof in \cite{Aly13} and Remark~8 in \cite{friz11} that $o( \frac{ \tilde{\nu}_\pm(y)^2}{2 y^2 {\tilde{\nu}_\pm}' (y) }  )$ has to be $\mathcal{O}(y^\frac{-3}{4})$. This will lead to an approximation of order $\mathcal{O}(y^\frac{-1}{4}) $ instead of $o(1)$ in (\ref{localvol_heston}).
 \end{remark}
 
 \begin{remark}
 The expansion (\ref{localvol_heston})  does not use the coefficients $d^\pm_i$ as they appear only on higher order terms which are not covered by Theorem~\ref{mgf_to_lv}. There will be therefore no need to calculate them if the only goal is to apply Theorem~\ref{mgf_to_lv}.
  \end{remark}

 A sharp asymptotic formula for the implied variance is easily obtained by applying  Theorem~\ref{from_mgf_to_iv}. The result is similar to the result obtained in \cite{achil14}.

 \begin{proposition}\label{iv_heston}
 The following expansion holds for the implied volatility under Heston's model
 \begin{equation}
 t  \sigma^2 (t,\pm k)  = 4 \tilde{\Lambda}^\ast_\pm (t,k) +2 \tilde{c}^\pm_t (k) \mp 2k - 4 \sqrt{ ( \tilde{\Lambda}^\ast_\pm(t,k) +\frac{  \tilde{c}^\pm_t (k)}{2})  ( \tilde{\Lambda}^\ast_\pm(t,k) +\frac{  \tilde{c}^\pm_t (k)}{2} \mp k)  }    + \mathcal{O} (k^\frac{-  1}{2} ) .
\end{equation}
where $\tilde{\Lambda}^\ast_\pm (t,k) = (\mu^\ast_\pm (t) - \frac{1}{  \tilde{\nu}_\pm (t,k)  } )k - \tilde{\Lambda} (t,\mu^\ast_\pm (t) - \frac{1}{  \tilde{\nu}_\pm (t,k)})  $, with  $\tilde{\nu}$ defined by (\ref{approx_nu}),   $ \tilde{\Lambda}$ is defined by
\begin{equation}
 \tilde{\Lambda}_\pm (t, \mu) = \frac{   \omega_\pm (t) }{  \mu^\ast_\pm (t) - \mu   } + \frac{2 a}{\sigma^2 }  \ln (    \frac{1}{ \mu^\ast - \mu}  ) + m_\pm (t),
\end{equation}
where $\omega_\pm$, $m_\pm$ are given in Proposition~\ref{expans_lambda_heston} and $ \tilde{c}^\pm_t (k) $ is given by
\[
 \tilde{c}^\pm_t (k) = - \ln ( k \frac{ \partial_k \tilde{\nu} (t,k)   }{\tilde{\nu}^2 (t,k)} ) + 2 \ln\left(\sqrt{  \frac{ \mu^\ast_\pm (t) }{2} }Ê- \frac{\mu^\ast_\pm (t) }{2} \sqrt{  \mp2 + 4 \mu^\ast_\pm (t) - 4 \sqrt{ \mu^\ast_\pm ( \mu^\ast_\pm \mp 1) }   }  Ê\right)
\]
 \end{proposition}
 
 \begin{proof}
 It follows from a direct application of Theorem~\ref{from_mgf_to_iv} by replacing $\Lambda$ with $\tilde{\Lambda}$. It's worth noticing that $\tilde{\Lambda}^\ast_\pm(t,.)$ is the Fenchel-Legendre transform of $\tilde{\Lambda}$.
  \end{proof}

    \subsection{Stein-Stein model}
The dynamics of the stock price in the Stein-Stein model is given by:
\begin{eqnarray*}
dS_t &=& \mu S_t d t + | Y_t| S_t  dW_t, \\
d Y_t &=& q (m - Y_t ) d t + \sigma dZ_t, ~~~~~Êd \langle W , Z \rangle_t = \rho d t.
\end{eqnarray*}
This  model was introduced by  Stein and Stein in \cite{Stein91}. In \cite{jacquier15}, Deuschel et al derive an asymptotic formulae for the density of $ X_t := \ln (S_t) $ as
\[
f(x)  =    e^{  -B_1(t)  x +B_2(t) \sqrt{x} - \frac{1}{2}  \ln (x) + B_3 (t) } ( 1 + \mathcal{O} (x^\frac{-1}{2})Ê ).
\]
From this we immediately have
\[
\mathbb{P} (X_t \geq x) = e^{  -B_1(t)  x +B_2 \sqrt{x} - \frac{1}{2}  \ln (x) + B_3(t)- \ln ( B_1 (t) )} ( 1 + \mathcal{O} (x^\frac{-1}{2})Ê ).
\]
We can then easily see that the (positive) moment generating function of $X_t$: $ \mathbb{E} e^{ \mu X_t} 1_{X_t > 0}$  explodes at $\mu^\ast_t = B_1$ and that
\[
\Lambda(\mu) := \ln \mathbb{E} e^{ \mu X_t}   = \frac{  \frac{1}{4} B^2_2  }{ B_1 - \mu} + \frac{1}{2} \ln (  \frac{ 1  }{ B_1 - \mu}) + B_3 - \frac{1}{2} \ln( \frac{B_2}{ 8 \pi B_1} ) + \mathcal{O} ( B_1 - \mu).
\]
We derive then a very similar expansion of the local volatility as well as the implied volatility (the positive wings) to the case of Heston's model by applying the formulas in Proposition~\ref{local_vol_heston} and Proposition~\ref{iv_heston}.

\end{document}